\documentclass{IEEEtran}

\usepackage{mathrsfs}
\usepackage{url}
\usepackage{graphicx}
\usepackage{latexsym}
\graphicspath{{./figs/}}

\newtheorem{lemma}{Lemma}
\newtheorem{theorem}[lemma]{Theorem}

\newtheorem{definition}{Theorem} 

\makeatletter
\def\theequation{\arabic{equation}}

\author{\IEEEauthorblockN{Qi Duan}\\
\IEEEauthorblockA{Carnegie Mellon University}
\\
\IEEEauthorblockN{Ehab Al-Shaer}\\
\IEEEauthorblockA{Carnegie Mellon University}
}
\title{Outsourcing SAT-based Verification Computations in Network Security }

\begin{document}
\maketitle

\begin{abstract}
The emergence of cloud computing gives huge impact on large computations. Cloud 
computing platforms offer servers with large computation power to be available for customers. These servers can be used efficiently to solve problems that are complex by nature, for example, satisfiability (SAT) 
problems. 
Many practical problems can be converted to SAT, for example, 
circuit verification and network configuration analysis.
However, outsourcing SAT instances to the servers may cause data leakage that can jeopardize system's security. Before
 outsourcing the SAT instance, one needs to hide the input information. 
One way to preserve privacy and hide information is to randomize the SAT  
 instance before outsourcing. In this paper, we present 
multiple novel methods to randomize SAT instances. We present a novel method to randomize the SAT instance, 
a variable 
randomization method to randomize the solution set, and methods to randomize
Mincost SAT and MAX3SAT instances. Our analysis and evaluation show the 
correctness and feasibility of these randomization methods. The scalability and generality 
of our methods make it applicable for real world problems. 
\end{abstract}




\section{Introduction}
\label{sec:intro}
\subsection{Motivation}
Cloud computing allows consumers and businesses to use applications and store 
large amount of data in cloud servers across the internet. It allows for much more efficient 
computing by centralizing storage, computational power and bandwidth.
It is convenient to outsource expensive computational tasks to cloud servers.

The main problem that makes users reluctant to outsource their computation is 
privacy preserving. Outsourcing computation may leak sensitive data that can 
put user's security on risk. Therefore, a user needs to ensure
 that the 
 data is secured before outsourcing. One way to solve the problem is to randomize the problem before
outsourcing. However, one needs to make sure that the randomization can be done 
efficiently and the randomized problem should not make the randmoized problem
much harder than the original
problem.

There are existing researches for privacy
 preserving data mining \cite{Verykios_04v,Clifton03,Agrawal00} 
and privacy preserving Linear Programming (LP) outsourcing \cite{WRW11}. 
However, Outsourcing 
the Satisfiability (SAT) problem is  also very important. 
SAT is one of the most fundamental problems in computer science and it
has broad  applications.
For example, SAT has important applications 
in circuit verification, software verification, task scheduling, etc \cite{MS08}.
SAT outsourcing is also very different from
LP outsourcing. LP can be solved 
in polynomial time and the algorithms of LP are mature and well known. 
One needs to outsource LP only if one has a very large instance. The hardness of 
LP comes from the size of the problem while the hardness of SAT is intrinsic 
in the problem itself, not only in the size of the problem. Hence the customers 
have strong motivation to outsource SAT  and the economic incentive for 
providing competitive SAT solvers is obvious. 
SAT is especially important in network configuration verification and planning. 
The complexity of network 
configuration verification and planning increases dramatically when the size of the network
 and the number of configuration
rules increase. It is reasonable for system administrators to outsource complicated
 network configuration verification and planning in the format of SAT. However, the SAT problems arising
 from network configuration and planning contain the configuration 
information that the system administrators 
do not want to leak. The solutions to the SAT problems may also contain 
the vulnerabilities or other sensitive system information. In this case 
security is the first concern to outsource SAT based configuration verification 
and planning. In some applications, multiple enterprise networks may need 
to carry out some computational tasks collaboratively. They need to verify that
the individual configurations will work for the  collaborative tasks. However, every
individual network is owned by a separate owner, and the owners may only want 
to reveal the interface information but not the internal information of their networks.
In this case the individual networks may randomize the configuration information and 
send the randomized configuration to a third party to verify the overall configuration
satisfies some global constraints.

In our proposed approach,
the steps for SAT outsourcing are as follows
\begin{enumerate}
\item The
user randomizes the SAT instance that he/she wants to outsource, using the
randomization tool and
sends it to the service provider.  
\item The service provider uses his/her algorithm to solve the randomized instance
 and returns the
solution. If the instance has no solution (unsatisfiable)
or the provider fails to solve it in some amount of time,  the provider should provide
the proof for the unsatisfiability or the proof that it really did the claimed amount of work.
 \item The user derandomizes the returned solution using
the derandomization tool and obtains the true solution
to the original problem. 
\item The user will validate the solution  returned by the service provider.

\end{enumerate}

  We should have an algorithm
to randomize SAT instances with the following requirements: \textit{first}, both the 
original and randomized instance must have the same satisfiability. \textit{Second}, 
any solution of the randomized instance can be efficiently
converted to the corresponding solution of the original 
instance. \textit{Third}, it should be computational hard for the
service provider to retrieve the original instance 
from the randomized instance. For the SAT instances
arising from configuration
verification and planning, the user may only
need to hide some of the statistic information of the
original instance, then we can relax this requirement that
it is  computational hard for the
service provider to retrieve these  statistic  properties. \textit{Fourth}, in 
some cases the user may also need
to randomize the relationship among the solutions
of the original instance. In this case we require 
that it is computational hard for the
service provider to figure out the relationship among 
the solutions of the original instance from the solutions of the
randomized instance except the number of solutions. 
For example, if the original instance has two solutions
$(0,0)$ and $(1,1)$, then the two solutions of the randomized instance
should not be complement to each other.   

These requirements will assure privacy preservation for the outsourced randomized SAT 
instances and it will also encourage users to outsource their SAT instances 
and benefit from third party facilities.  
 
The main objective of this project is to provide the randomization/derandomization tool
for the client who want to outsource SAT-based verification. We provide multiple
randomization algorithms and the user can choose an appropriate 

The most straightforward method to randomize the 
SAT instance is to permute the index of the variables or flip the
true/false  of the appearance of the variables. It is not trivial
for the provider to differentiate two isomorphic
SAT instances since it is not known if there exists a polynomial time algorithm
 for graph isomorphism \cite{Schoning88}. However, merely permute
the index of the variables or flip the variables' truth/false appearance  cannot change
the relationships among the solutions and the statistic 
properties of the instance. 
 The work in \cite{Dimitriou22} is a general  privacy-preserving obfuscation for outsourcing SAT formulas but its performance is not shown for network security related problems such
 as firewall analytics.

It can also be shown that there
is much space for improvement for current
SAT solvers.    Even a relatively small instance with thousands of variables
may be beyond the ability of the best SAT solvers 
today. Table~\ref{tab:3sat1} shows the time to solve a random 
instance with $n$ variables and $m$ clauses 
with zChaff \cite{zchaff} SAT solver. We can see that 
the time to solve a 3SAT instance increases dramatically when 
So we can see that current SAT solvers
are not efficient enough for many applications.  
There is enough motivation for users to outsource
SAT based computation, and for the cloud service providers
to develop competent SAT solvers or applications
that contain SAT solvers.
 
\begin{table*}[ht]
\begin{center}
\begin{tabular}{| l | l | l|l|l|l|}

 \hline
  n&m & Time to solve (s) & n&m & Time to solve (s) \\
  \hline
  200 & 900 & 0.91  & 400 & 1500 &  $<0.01$  \\\hline
   300 & 1200 & $<0.01$ & 400 & 1600 &  1.39  \\ \hline
  300 & 1250 & 0.12 & 400 & 1650 &  1414  \\ \hline
   300 & 1300 & 395 & 400 & 1700 &  16337  \\ \hline

\end{tabular}
\end{center}
\caption{Time to solve the 3SAT instance}
\label{tab:3sat1}
\end{table*}

Our contributions in this paper come in presenting several methods to randomize SAT instances as follows: \textit{first}, 
a method to randomize some statistical properties of a SAT instance
by noise injection. \textit{Second}, a method to randomize the whole structure of a SAT instance. \textit{Third}, a method to randomize a solution set. \textit{Fourth}, methods to randomize Mincost SAT and MAX3SAT.
We also study an important practical example of outsourcing SAT based configuration verification,
that is firewall equivalence verification. To the best of our knowledge, this is
the first work to investigate privacy preserving in SAT outsourcing. 

The rest of the paper is organized as follows.
Section~\ref{sec:adv} discusses the computation model,
adversary model and requirements for SAT outsourcing.
Section~\ref{sec:ran} describes the 
methods to randomize SAT instances before outsourcing.
 Section~\ref{sec:case} presents  the
case study of firewall equivalence verification.
Section~\ref{sec:eva} shows  the
evaluation results. 
Section~\ref{sec:related} presents  the
related works.
Section~\ref{sec:legal} discusses the legal implications of 
SAT outsourcing and
section~\ref{sec:conclusion}
concludes the paper and presents directions for future work.

\section{Computational Model, Adversary Model and Requirements for SAT Outsourcing}
\label{sec:adv}

\subsection{Computational Model}
In the computation model of SAT outsourcing, there are two participants. The first
participant is the user, who wants to outsource his/her SAT problem. The 
second participant is the cloud service provider.
The steps of SAT outsourcing are as follows:

\begin{enumerate}
\item The
user randomizes the SAT instance that he/she wants to outsource and
sends it to the service provider.  
\item The service provider uses his/her algorithm to solve the randomized instance
 and returns the
solution.
\item The user derandomizes the returned solution and obtains the solution
to the original problem. 
\end{enumerate}
 
\subsection{Adversary Model}
In our discussion of this paper, we consider three types of service providers:
\begin{itemize}
\item \textit{Honest providers.} Honest providers always report the answer from
an honest execution of their SAT algorithm.

\item \textit{Lazy providers.} Lazy providers may report ``fail" for
an instance without executing their SAT algorithm. Since
the user also needs to pay for the provider if the
user cannot present  evidence for cheating behavior of the provider, the
provider can benefit from lying. 

\item \textit{Malicious providers.} A malicious
provider may have two kinds of malicious behaviors. He may try to figure out the 
original instance from the randomized instance or he may also report ``unsatisfiable"
even if the instance is satisfiable. To do this, the malicious provider may provide
a wrong unsatisfiable core for the user, or cheat in replying
the user's questions about the instance during
the interactive or non-interactive proof procedure for the
unsatisfiability of the instance. The malicious provider may
use the solution of the SAT instance  to launch attacks
or provide the solution to third parties.

\end{itemize}

One needs to detect malicious providers and lazy providers for outsourcing SAT
 instances. We should have an algorithm
to randomize SAT instances with the following requirements: \textit{first}, both the 
original and the randomized instance must have the same satisfiability. \textit{Second}, 
any solution to the randomized instance can be efficiently
converted the a corresponding solution to the original instance. \textit{Third}, it
 should be computationally hard for the
service provider to retrieve the original instance 
from the randomized instance. For the SAT instances
arising from configuration analysis and
verification, the user may only
need to hide some of the statistical information of the
original instance, then we can relax this requirement that
it is  computationally hard for the
service provider to retrieve these statistical  properties. \textit{Fourth}, in some cases the user may also need
to randomize the relationship among the solutions
of the original instance. In this case we require 
that it is computationally hard for the
service provider to figure out the relationship among 
the solutions of the original instance from the solutions of the
randomized instance except the number of solutions. 
For example, if the original instance has two solutions
$(0,0)$ and $(1,1)$, then the two solutions of the randomized instance
should not be complement to each other.   

The above requirements will assure privacy preservation for the outsourced randomized SAT 
instances and it will also encourage users to outsource their SAT instances 
and benefit from third party facilities.   

\subsection{Classification of Outsourcing Security}
Informally, we say that a user or client $C$ securely outsources some work to
cloud service provider $S$, and $(C,S)$ is an
outsource-secure implementation of a cryptographic algorithm Alg if (1) $C$ and
$S$ implement $Alg$, such that $Alg = C^S$ and (2) $S$
 cannot learn the sensitive information about the input and
output of the computation. 

 In the following, we introduce the formal definitions for secure
outsourcing. We adapt the  definition from \cite{CLMTL12}, with some modifications.

\begin{definition} \textbf{Full Outsourcing-Security}
Let $Alg$ be an algorithm with outsource
input/output. A pair of algorithms $(C,S)$ is said to be an outsource-secure implementation
of $Alg$ if:
1. Correctness:  $C^S$ is a correct implementation of $Alg$.
2. Security: For all probabilistic polynomial-time adversaries $A = (E,S')$, where $E$ is the
adversarial environment that submits adversarially chosen inputs to $Alg$,  there
exist probabilistic expected polynomial-time simulators $(S_1, S_2)$ such that the
random variables obtained from the view of the
input/output of $Alg$ and the view from the
 execution of the simulators are computationally indistinguishable.
\end{definition}

Note that this is the strongest form of security, which means the adversary can learn
nothing from the input/output of the algorithm. 

If the client only cares about the privacy of the original instance but not the
the privacy of the solution, the definition can be modified to be 
that the view from the input of the algorithm is computationally indistinguishable
from the view of any other input which has the same set of solutions.

 \begin{definition} (\textbf{Instance-privacy Outsourcing})
A pair of algorithms $(C,S)$
is said to be an instance-privacy outsourcing  of $Alg$ if (1) 
$C^S$ is a correct implementation
of $Alg$ and (2) $\forall$ inputs $x$ is computationally indistinguishable
from the view of any other input $x'$ which has the same set of solutions as $x$.
\end{definition}

\section{Randomizing SAT Instances} 
\label{sec:ran}
In this section we present  methods that can be used to randomize SAT instances and prepare them for outsourcing.
\subsection{Permutation of Variables and Negation Flipping}
The most straightforward method to randomize a  
SAT instance is to permute the index of the variables or flip 
true/false  values of the variables. It is not trivial
for the provider to differentiate between two isomorphic
SAT instances; since it is not known if there exists a polynomial time algorithm
 for graph isomorphism \cite{Schoning88}. However, merely permuting
the index of the variables or flip the variables' true/false values cannot change
the relationships among the solutions and the statistical 
properties of the instance.

\subsection{Matrix Multiplication Randomization}

The noise injection method for SAT randomization can only
hide some of the statistical properties of the original
instance. If we want to completely randomize all  
information of the original instance except the solution set,
 we can use a more complicated method called
matrix multiplication randomization. The method
has significant overhead.
If the requirement of privacy preservation is high, the user
may choose this method. 

Here we only consider 3SAT problem, since every SAT instance
can be easily converted to a 3SAT instance.

The following discussion shows how to convert 
a 3SAT instance into a matrix form and how to randomize the generated matrix. We can first convert the
3SAT instance to an equation array of 0/1 linear constraints.
Inequalities can be converted to equalities by adding dummy variables.
After this procedure we can multiply a random 0/1 matrix in both sides of the equation array and
now the problem is converted to a 0/1 linear constraint satisfaction problem. 
Any solution to the new linear integer programming instance will be
a valid solution for the original SAT instance. 
 
The detailed steps are as follows:

\begin{enumerate}

\item \textit{Convert to equation}: For 
every  variable $x_i$ in the original 3SAT, create a corresponding variable
$y_i$ in the
  created 0/1 linear constraint satisfaction instance. For every clause in the 3SAT
instance, convert it to an equation with two dummy variables. Suppose the
original clause is $x'_{i1} \lor x'_{i2} \lor x'_{i3}$, where 
$x'_{ij}$ may be variable $x_{ij}$ or its  negation form $\overline{x_{ij}}$,
 ($1\leq j \leq 3$), then we
 convert it to the following equation:
\begin{equation} 
\label{equ:convert}
y'_{i1} + y'_{i2} + y'_{i3} + y_{d1} +  y_{d2} = 3.
\end{equation}

Here $y'_{ij}$ is $y_{ij}$ if $x'_{ij}$ is $x_{ij}$, and is $(1- y_{ij})$ if $x'_{ij}$ is
$\overline{x_{ij}}$ ($1\leq j \leq 3$). There are the two dummy variables $y_{d1}$ and $y_{d2}$ 
for the clause.
   Now we get the equation set $AX = B$ where $A$ is an $m$ by $n+ 2m$ matrix. Here
  $m$ is number of clauses, $n$ is number  of variables in the original 3SAT instance.

\item \textit{Random matrix multiplication}: Generate a random $m$ by $m$ 0/1 matrix $R$ with full rank,
 and    multiply $R$ to both sides of equation array. Now we have $RAX= RB$.
Note that here $RA$ is still a   $m$ by $n+ 2m$ matrix, $RB$ is a $m$ by 1 matrix.

\item \textit{Outsource the problem}: Send this 
equation array to the service provider, and ask it to solve the
equation array as a 0/1 linear constraint satisfaction problem.

\end{enumerate}

The following example shows the details of converting a 3SAT instance to a matrix representation. Consider the following 3SAT instance: \[ (x_1 \lor x_2 \lor x_3) \land  ( \overline{x_1} \lor x_2 \lor \overline{x_3}). \]
We need to add 4 dummy variables $x_4,x_5,x_6,x_7$, two for each clause. The 
converted 0/1 linear  constraint satisfaction instance is:
 \[  x_1+x_2+x_3 +x_4 + x_5 = 3 \]
\[  - x_1  +x_2 -x_3 +x_6 + x_7 = 1. \]
We can see that any solution of the original 3SAT instance can also be converted to a solution
for the new 0/1 linear constraint satisfaction instance if we set the Boolean \textit{true} value to the integer value
1 and the Boolean \textit{false} value to the integer value 0, and set the values of the dummy variables as follows:
 \begin{itemize}
\item If the clause is satisfied by all three of the literals, set the
two dummy variables correspond to the clause to  0.

\item If the clause is satisfied by exactly two of the three of the literals, set 
one of the two dummy variables correspond to the clause to  0, another one to 1.
  
\item If the clause is satisfied by exactly one of the three of the literals, set both
  dummy variables correspond to the clause to   1.
\end{itemize}

Next we prove that any solution to the original 3SAT instance
can be converted to a solution to the new problem. Conversely, any 
solution to the new problem instance can also be converted 
to a solution to the original problem.

\begin{theorem}
Any solution to the randomized 0/1 linear constraint satisfaction instance 
can be converted to a solution in the original SAT instance 
if we set the integer value 1 to Boolean true    and 0 to  false.
Any solution of the original SAT instance also
corresponds to a solution in the randomized 0/1 linear constraint satisfaction  instance.  
\end{theorem}

\begin{proof}
For any solution to the original instance, the  assignment
of the variables will satisfy any clause. That means for a clause
$(x_{i1} \lor x_{i2} \lor x_{i3})$, one of the $x_{ij}$ ( $ 1\leq j \leq 3$)
must be true. If $k$ ($ 1\leq k \leq 3$)  variables
in the clause are satisfied, we can set $3-k$ dummy variables
correspond to the clause to be 1 in   Equation~\ref{equ:convert}.
This means that we have a solution for equation set $AX=B$, consequently,  
we also have a solution for   equation set $RAX=RB$.
On the other hand, any solution for $RAX=RB$ is also a solution
for $AX=B$ because $R$ is invertible. Then for every clause,
one of the $y'_{ij}$ ($ 1\leq j \leq 3$) in   Equation~\ref{equ:convert} must be one, which means
the corresponding clause in the original SAT instance can be satisfied.

\end{proof}

\begin{theorem}
The randomized matrix $RA$ can be any matrix that satisfies
the same column vector linear relationship as that of matrix $A$. This means
the outsourcing method keeps the instance privacy.
\end{theorem}

\begin{proof}
After adding the dummy variables, $A$ will have rank $m$. So the number of
linear independent column vectors in $A$ is $m$. If we 
take the $m$ by $m$ matrix $A_1$ that contains the $m$  linear independent column vectors
of $A$, then for any full rank $m$ by $m$ matrix 
$A_2$, we can choose matrix $R=A_2A_1^{-1}$. Now we can see that
$RA$ will contain all column vectors of $A_2$. Other
columns of $RA$ will be linear combination of column vectors of $A_2$. 
This shows that the randomized matrix $RA$ can be any matrix that satisfies
the same column vector linear relationship as that of matrix $A$.
\end{proof}

 The properties of this technique are: \textit{first}, the transformation can be done efficiently. Here we need
only the matrix multiplication in the transformation. \textit{Second}, the old problem and the new 
problem have similar hardness in theory.
 Since any solution to the new problem is also a 
solution to the old
instance, many existing search based algorithms that work for SAT
 will also work for the 0/1 linear constraint satisfaction problem.
\textit{Third}, this technique provides  a 
complete randomization of the structure of the original
instance. The choice of $R$ can be arbitrary,
so there is no way to recover the original instance without knowing $R$.

\subsection{Solution Set Randomization}
\label{sec:var_random}
The randomization method in the previous section converts the original SAT instance
to a randomized instance with the same solution set. In some
scenarios, one may require to randomize the solution set of the SAT instance. Here we consider 
a method to randomize the solution set  of the original SAT. 

Suppose there are $n$ variables $x_1, \ldots, x_n$ in the original SAT instance. We
 set $X=[x_1,\ldots,x_n]^T$. We can generate a random full rank $n$ by $n$ 0/1 matrix
$R$, and define new variable vector 
  \[ Y=[y_1, \ldots, y_n]^T=RX. \]

Note that matrix multiplication is done in finite field $F_2$.   
Now we have $X=R^{-1}Y$, which means that every $x_i$ ($1\leq i \leq n$) is the exclusive or
of some of $y_i$s. For every clause $x_{i1}\lor x_{i2}\lor x_{i3}$, we
can replace $x_{ij}$ ($1\leq j \leq 3$) with the corresponding $y_i$s, and 
convert the new Boolean formula to standard 3CNF formula.
As an example, suppose $x_{i1}=y_1 \oplus y_2 $, $x_{i2}=y_3 \oplus y_4$,
$x_{i3}=y_5 \oplus y_6$, then $x_{i1}\lor x_{i2}\lor x_{i3}$ is equivalent
to
\begin{eqnarray}
   \nonumber (y_1 \land \overline{y_2}  ) \lor (\overline{y_1} \land y_2 ) \lor  
    (y_3 \land \overline{y_4} ) \\ \nonumber \lor (\overline{y_3} \land  y_4) \lor (y_5 \land \overline{y_6}  )
           \lor (\overline{y_5} \land  y_6),  
\end{eqnarray}

which can be converted the following CNF
\[  (\displaystyle\bigvee_{1 \leq i \leq 6}z_i )\bigwedge 
   \left(\displaystyle\bigwedge_{1 \leq i \leq 6}Z_i \right), \]

where 
 \begin{eqnarray}
 \nonumber
 Z_i=  (\overline{z_i} \lor y_i ) \land
   (\overline{z_i} \lor \overline{y_{i+1}} ) 
   \land (z_i  \lor \overline{y_{i}} \lor y_{i+1}) 
\end{eqnarray}
when $i=1,3,5$, and 
\begin{eqnarray}
\nonumber
 Z_i=  (\overline{z_i} \lor y_i ) \land
   (\overline{z_i} \lor \overline{y_{i-1}} ) 
   \land (z_i  \lor \overline{y_{i}} \lor y_{i-1}) 
\end{eqnarray}
when $i=2,4,6$.

Here $z_i$ ($1 \leq i \leq 6$) are dummy variables.
 
For variables that are exclusive or of more than two old variables, we can also add dummy
variables to convert it to 3CNF. For example,
 $x_{i}=y_1 \oplus y_2 \oplus y_3 $ can be converted to
\begin{eqnarray}
\nonumber
 x_i= (z \lor y_1 \lor y_2 ) \land (z \lor \overline{y_1} \lor \overline{y_2} ) \\
\nonumber
    \land ( \overline{z} \lor \overline{y_1} \lor y_2 ) \land (\overline{z} \lor y_1 \lor \overline{y_2} )  \\
\nonumber
  \land (z \lor \overline{y_3} ) \land (y_3 \lor \overline{z} ), 
\nonumber
\end{eqnarray}
 where $z$ is a dummy variable.

In this way we can convert the SAT instance to a new SAT instance, and the solutions
to the old instance can be recovered from the solutions to the new instance
by  $X=R^{-1}Y$.  The relationship among the solutions  in the original SAT instance
 will be randomized. 
The shortcoming of this approach is that the number of clauses in the new instance
 will increase with a factor of the number of variables. To reduce the complexity, we can
 use a sparse matrix $R$ in the randomization.
If the number of non-zero entries in $R$ is linear in $n$, then the number of
clauses will be linear in $n$.
After one randomizes the solution set, one can use the noise injection
method or the matrix multiplication method to randomize
the instance.

\subsection{Randomizing Mincost SAT}
\label{sec:mincost}
Mincost SAT \cite{Papadimitriou93} is an important
variant of SAT.
We can use the  noise injection or the matrix multiplication method to 
randomize any  Mincost SAT instance since the cost of a solution
for the randomized instance is also the cost of the corresponding solution for
the original instance. Since the user needs to set the cost of
all dummy variables to be 0, he/she may reveal the 
dummy variables when the cost function is provided to the service provider.
To deal with this situation, we can convert the
cost function to a Boolean circuit with variables
in the original SAT instance. Suppose in the original
instance there
are $n$ variables $x_1, \ldots, x_n$, and the 
cost function is 
\[ C = c_1x_1 + c_2x_2 + \ldots + c_nx_n, \] 
where $c_i$ ($1 \leq i \leq n$) is the cost of variable $x_i$ and  $x_i$
 takes 0/1 values. 
The Boolean circuit $C_1$ have $O(n\beta)$ gates, $O(n\beta)$ variables,
and $\beta$ output bits $b_1, \ldots, b_{\beta}$,
where $\beta$ is the number of bits in the representation of $c_i$  values
and $b_1$ is the most significant bit.

We can convert $C_1$ to a  CNF $C_2$ and generate a CNF $C_3$ which 
 combines $C_2$  with the original CNF, and set the
new cost function to be
\begin{equation}
\label{eqn:cost}
 C' = 2^{\beta}b_1 + 2^{\beta-1}b_2 + \ldots + b_{\beta}.
\end{equation}
Now we can randomize $C_3$ and  
the provider cannot distinguish the dummy variables and non-dummy variables anymore.

\subsection{Randomizing MAX3SAT}
In this section we present a method to randomize MAX3SAT instances. 
 For every clause $x_{i1} \lor x_{i2} \lor x_{i3}$
 in the original SAT instance, the user can
create a new variable $y_i$ and the following formulas 
\begin{equation}
\label{eqn:clause}
 y_i \Leftrightarrow (x_{i1} \lor x_{i2} \lor x_{i3}), \,\, 1\leq i \leq m,
\end{equation}

where $m$ is the number of clauses.
Next the user can combine all these formulas together with the original
CNF to get a new Boolean formula $C_1$, and  convert $C_1$
to an equivalent CNF  $C_2$.
The objective  of the original  instance becomes

 \textit{Maximize} $y_1 + y_2 + \ldots + y_n$,

 where $y_i$ ($1\leq i \leq n$) takes 0/1 values.

Now the user can convert the original problem to a new Mincost SAT problem. The
CNF in the new problem is  $C_2$, 
and the new objective function is: 

\textit{Minimize} $- y_1 - y_2 - \ldots - y_n$.

The user can also use the  noise injection or the matrix multiplication method to 
randomize the CNF and use the method in Sec.~\ref{sec:mincost} to hide the
dummy variables.

\subsection{Verification of the Correctness of the Result}
For any SAT instance, the provider may return three types of results. The first type of results
is one or multiple correct solutions. The second is \textit{``unsatisfiable"} with
the related unsatisfiable core (or proof of the unsatisfiable core), and the
third is \textit{``fail"}; this happens when the provider cannot determine the
satisfiability of the instance in a specified amount of time. 
When the provider returns one or multiple solutions, the user can
verify the results easily. 
The latter two cases are 
difficult to verify.   
In computational theory, 
it is  believed that 
the complexity class $Co\mbox{--}NP$ is unlikely to be in class $NP$.
So one cannot provide a  polynomial time verifiable proof
for unsatisfiable SAT instances. For some instances, the user can guarantee 
that the instance is satisfiable. For example, if the
user want to outsource integer factoring or discrete logorithm
by converting the problems to SAT, he/she knows
that some solution must exist, and the conversion is simple (to convert
factoring to SAT, one just needs to convert the multiplication circuit
to SAT, which can be done in $O(n^2)$ time, where $n$ is the number of bits
of the integer. Discrete logorithm can also be converted to SAT efficiently).
For these SAT instances, the service provider cannot cheat with the result \textit{``unsatisfiable"}.

It is also unlikely to design practical interactive or non-interactive proofs
for unsatisfiable SAT instances. 
By Shamir's theorem \cite{Shamir92}, any problem in $PSPACE$
can be verified with interactive proofs in polynomial time. However, the proof
of Shamir's theorem
works only from a theoretical perspective because it assumes that the prover
has infinite computational power, which is not practical for existing
service providers. 
The existing techniques used
to defeat service provider cheating cannot be applied for
verification of unsatisfiable SAT instances. The 
method presented in \cite{grid} uses Merkle  hash 
tree commitment for computation verification \cite{Merkle89, Merkle80}, 
and the work in \cite{GM01} combines some pre-computed results
 with the computation workload to detect lazy providers.
Both of them can only  check the cheating behavior
in a non-negligible portion of all possible computation branches,
but in the verification of unsatisfiable SAT instances one needs
 to verify the correctness of every possible computation branch.

In the case that the provider reports ``fail" 
and the user wants to verify that 
the provider  has really spent the claimed amount of time on the
problem, 
the provider can build the Merkle hash tree for the computational
procedure (such as the searched branches) and use the similar method in \cite{grid} to verify
the correctness of the tree.
The verification
takes $O(logn')$ time in communication and $O(logn')$ computation
overhead for the user where $n'$ is the size of the tree.
One problem with this approach
is that the user may obtain the details of the algorithm
from the verification procedure,
and the algorithm may be the secret of the provider.

For some unsatisfiable SAT instances 
 the provider may
find the unsatisfiable core \cite{GN03}. The provider can send back this core
along with  the proof of the core
to the user and the user can verify the 
correctness of the core. However the verification
of the core may be beyond the computational power of the user. 
In this case, the customer  can 
 outsource different randomized versions of the 
core to several other service providers. If all these
providers answers ``unsatisfiable" or ``fail"  for  the  unsatisfiable core, 
the user accepts the result. As long as one provider returns
a solution for the  unsatisfiable core, the first provider is caught with cheating
and will lose credibility.

To avoid the case when the SAT is solvable but the
provider simply report ``fail", one can also send  different randomized
versions of the original to different service providers.  If one
of the provider can find one satisfiable assignment for it, then the user can show
the solution to other providers
that reported ``fail". This solution can be easily
verified. In this case, the user will only need to pay the full charge for the
provider that reports a valid solution. If none of the providers
report a valid solution, it means  the instance is really hard,
and the user will pay full charge to all providers. The providers 
receive different randomized versions of the problem, so they cannot
collude since they cannot determine whether they receive the same original instance or not.
A provider may still succeed in cheating in the case that the problem
happens to be hard and nobody else can solve it. But if
the provider is caught with laziness or incompetency in solving the problem,
he may lose his credibility and future users. The providers have enough motive to
work ``hard" to solve problems since it may get more compensation than
to be ``lazy".

\subsection{Outsourcing Multi-party SAT-based Computation}
In applications that multiple partners jointly execute some tasks, the
multiple partners need to verify that the configurations of their networks
are correct for the joint taks. However, the partners may 
only want to reveal the interface information (the inferface between the
partner's network and other partners), and they may not want to reveal their
internal configuration information. In this case one must find a way to
carry out secure multi-party computation. Existing protocols for secure multi-party
computation are too expensive and not practical for real application.  

Suppose there are $n$ partners and the
 configuration properties that need to be verified can be represented
as a Boolean formula 
\[ P = f(B_{11}, B_{1u_1}, \ldots, B_{n1}, B_{nu_n}) \] 

where $B_{i1}, \ldots, B_{iu_i}$ are the Boolean formulas that only
involve the configuration of network of partners $i$. 

We can use the following procedures to randomize formula $P$ before outsourcing
the verification task:

\begin{itemize}
\item{Every partner $i$ convert $B_{i1}, \ldots, B_{iu_i}$ 
to CNFs and randomize them to $B'_{i1}, \ldots, B'_{iu_i}$.}
\item{Every partner $i$ sends $B'_{i1}, \ldots, B'_{iu_i}$ to 
a third party or a representative selected among them.}
 \item{The partners agree with a public key using some
key generation protocols and send the key
to the third party or the representative.}
\item{The third party or the representative generates
a random matrix using the key as the seed and 
randomize $P$ using the random matrix.}
\item{The third party or the representative sends randomized formula $P'$
to the cloud service providers.}
\end{itemize}

Note that the individual configuration information related to every partner is
randomized at the first step of the above procedure.

\section{Case Study: Firewall Equivalence Checking}
\label{sec:case}

Firewalls are the most important  network access control devices
 that control the traversal of packets
across the boundaries of a secured network based on. 
A firewall policy is a list of ordered filtering rules that define the actions
performed on
matching packets. A rule is composed of filtering fields (also called header
tuples) such as protocol type, source IP address, destination IP address, source
port and destination port, and an action field. Each rule field could be a single value or
range of values. Filtering actions are either to accept, which passes the packet
into or from the secure network, or to deny, which causes the packet to be discarded. The
packet is accepted or denied by a specific rule if the packet header information
matches all the fields of this rule. Otherwise, the following rule is examined and
the process is repeated until a matching rule is found or the default policy action is
performed.
The filtering rules may not be disjoint, thereby packets may match one or
more rules
in the firewall policy. In this case, these rules are said to
be dependent or overlapping and their relative ordering must be preserved
for the firewall policy to operate correctly. 
 
If  two firewalls have large rule sets and
the network administrator want to verify if they
are equivalent, then he/she
may outsource the verification task to some 
service provider.

The most straightforward method is to use random mapping 
to randomize configuration policy rules for outsourcing.
 For every blocks in the  IP, one can generate a mapping
from 0-255 to 0-255. For the port numbers, one can also have a mapping
  from 0-25535 to 0-25535. Note that the mapping should be preserved for all 
rules. For example, consider a firewall policy with two rules 
shown in Table~\ref{tab:rulemap}. Based on the mapping shown 
below, the randomized rules is shown in Table~\ref{tab:mappedrule}.

\begin{table}
\centering
\begin{tabular}{ |c|c|c|c| }
\hline
   src IP  &    src port &    dest IP &       dest port  \\ \hline
  10.11.12.*  &     100 &    10.14.15.* &       80 \\\hline
  152.15.10.*  &    99  &   152.15.*.* &       80\\\hline
\end{tabular}
\caption{The original rules}
\label{tab:rulemap}
\end{table}
	
 \begin{table}
\centering
\begin{tabular}{ |c|c|c|c| }
\hline
   src IP  &    src port &    dest IP &       dest port  \\\hline
  23.170.55.*   &     471 &    23.76.142.*  &       2313 \\\hline
  163.201.97.* &    15717   &   163.201.*.*  &      2313 \\\hline
\end{tabular}
\caption{The randomized rules}
\label{tab:mappedrule}
\end{table}

Suppose for  IP block 1, the random mapping is:
 \[  10 \leftrightarrow 23, 152 \leftrightarrow 163, 100 \leftrightarrow 41 \]
 For  IP block 2, the random mapping is:
 \[ 11 \leftrightarrow 170, 14 \leftrightarrow 76, 15 \leftrightarrow 201  \]
 For  IP block 3, the random mapping is: 
 \[  12 \leftrightarrow  55, 15 \leftrightarrow  142, 10 \leftrightarrow  97  \]
 For port number,  the random mapping is: 
 \[ 100 \leftrightarrow 471, 99 \leftrightarrow 15717, 80 \leftrightarrow 2313  \]

    To randomize the rules in this way, the  IP and port numbers are hidden and
the semantics of the rules can  be maintained.  However the service provider can still get the entropy
information. For example, if port 80 appears frequently in the rules, the mapped
number will also appear frequently. The service provider may deduce the mapping from the
statistics of the field values of rules.  So we can see that this kind of 
naive randomization method is not enough work for user privacy. We need to
seek more sophisticated approach for this problem.

We can use the SAT randomization methods in \S~\ref{sec:ran} to randomize the
firewall rule sets.
To do this, we  need to represent   a firewall as a Boolean formula $F$. 

Suppose a firewall contains $u$ rules $r_1, r_2, \ldots, r_u$,
we can denote the Boolean formula corresponding $r_i$ as $A_i$ ($1 \leq i \leq u$).
For  every single  rule in the firewall, we need 16 bits to represent source port and destination port, 32 bits
to represent source and destination address. In total we need 96 bits $b_1, b_2, \ldots, b_{96}$.
If the action of the rule $r_i$ is accept, then the rule 
can be represented as
\begin{eqnarray} 
 A_i \Leftrightarrow (b_{i1} \land b_{i2} \land \ldots \land b_{ik}),
\label{eqn:accept}
\end{eqnarray}
 
where the bits $b_{i1}\ldots b_{ik}$ are the corresponding bits of the field in the rule.
Here  $k$ is the number of bits needed 
to represent a single rule, which is 96 in this case.

If the action of the rule is deny, then the rule 
can be represented as
\begin{eqnarray}
 A_i \Leftrightarrow (\overline{b_{i1}} \lor \overline{b_{i2}} \lor \ldots \lor \overline{b_{ik}} ). 
\label{eqn:deny}
\end{eqnarray}

 If rule $r_i$ is independent from all other rules, then we can
add it into $F$ as
  \[ F = F \lor A_i \]
if the action of $r_i$ is accept, and
 \[ F = F \land A_i \]
if the action of $r_i$ is deny. All independent rules can be added in this way.   

Next we consider the remaining rules. Without of loss of generality,
we assume the remaining set of rules is $R' =\{r_1, r_2, \ldots, r_u' \}$ ($u'\leq u$),
and the Boolean representation of a single rule $r_i$ ($1 \leq i \leq u'$)
in $R'$ is $A'_i$.

Now the Boolean formula that represents those dependent rules can be represented as
 \begin{eqnarray} F' =  A'_1 \lor (\overline{A'_1} \land A'_2) \lor \ldots \lor 
   ( \displaystyle \bigwedge_{1 \leq i \leq u'-1}\overline{A'_i} \land A'_{u'}).
\label{eqn:whole}
\end{eqnarray}

The whole firewall can be represented
as $F \lor F'$.

Suppose the Boolean formulas that represent two firewalls are $F_1$ and $F_2$,
the non-equivalence of the two firewalls   
is equivalent to the satisfiability 
of the formula
\[ (F_1 \lor F_2) \land ( \overline{F_1} \lor \overline{F_2}).  \]

We can convert this formula to the standard CNF representation \cite{Papadimitriou93}. 
The number of variables and clauses in the standard CNF formula is 
linear in the size of the original formula.
In the worst case,
the total number of variables in the CNF representation 
of the firewall equivalence is at most $O(u^2 + ku)$ and the total number of clauses
is also $O(u^2 + ku)$.   
To randomize the resulting 3CNF, we can use the randomization methods in Sec.~\ref{sec:ran}.


\section{Evaluation}
\label{sec:eva}
We randomly generated SAT instances to evaluate the outsourcing
techniques. Every literal in every clause of the SAT instance 
is chosen uniformly from  the set of variables. All evaluations
are done in a computer with dual core 1.6G Pentium IV processor. 
We used the zChaff \cite{zchaff} SAT solver to solve 
SAT instances  and
 Yices \cite{yices} to solve 0/1 linear constraint satisfaction instances.
 Yices is an SMT (Satisfiability Modulo Theories) \cite{DP60} solver
which can be used to solve  constraint satisfaction problems
in many diverse areas.

{\bf Feasibility of Matrix Multiplication Method (satisfiable instances):}
Table~\ref{tab:sat1} shows the time to randomize the original 3SAT instance, the
time to solve the original 3SAT instance by  zChaff and
 Yices, and the time to solve the randomized 3SAT instance
by Yices. 
$m$ and $n$ are  the number of clauses and variables, respectively. In all the 3SAT
instances in this table, we have  $m/n=3$ or $m/n=4$, where
the instances with  $m/n=4$ is harder than the instances with $m/n=3$
because $m/n=4$ is more close the
phase transition value of 3SAT \cite{phase}.
First we note the performance of 
Yices is much worse than zChaff for 
SAT instances. This is because Yices is not designed
for SAT, and it uses a more complicated data structure 
than zChaff. However we believe this can be improved 
in future SMT solvers (to directly use existing efficient
SAT algorithms when the instance is a pure CNF formula).  
The time to solve the randomized instance
with Yices is also much larger than the time
to solve the original instance in Yices. This is because
the randomization procedure introduces a large 
number of dummy variables and every variable may appear in every
linear constraint. Though the price for
randomization is significant, the matrix multiplication randomization method
can still be applied within practical limits to the cases when the user want absolute
privacy for the original instances. Yices and other existing linear integer programming
tools are not designed specifically to solve 0/1 linear constraint satisfaction
problems. We believe that there is much room
to improve the efficiency to solve 0/1 linear constraint satisfaction
problems in the future. 
\begin{table*}[ht]
\begin{center}
\begin{tabular}{|c|c|c|c|c|c|}

 \hline
  n&m & Randomization time & Time for original(zChaff) & Time for original(Yices)&
  Time for randomized(Yices)\\
  \hline
  100 & 300 &  0.27 & $<0.01$ & $<0.01$ & 1.67  \\\hline
  100 & 400  & 0.8 & $<0.01$ & $<0.01$ &  2.33 \\\hline
  300 & 1000 & 19.74 & $<0.01$   & 0.02 & 14.20  \\ \hline
  300 & 1200 & 34.04 & 1.71 &   21.36 & 8391.28  \\ \hline
  500 & 1500 & 70.15 & $<0.01$   & 0.03 & 31.35   \\\hline
  1000 & 3000 &597.43 & $<0.01$  & 0.06 & 92.78  \\\hline
\end{tabular}
\caption{Matrix multiplication overhead and cost to solve satisfiable instances (seconds)}
\label{tab:sat1}
\end{center}
\end{table*}

\section{Related Works}
\label{sec:related}
The first research for secure outsourcing expensive
computations was Yao's garbled circuits \cite{Yao82}.
Gentry's work on Fully Homomorphic Encryption (FHE) \cite{Gentry09} 
showed that it is possible to achieve secure computation
outsourcing in theory. Gennaro et al. \cite{Gennaro:outsourcing} presented a work to outsource computations to untrusted workers. A fully-homomorphic encryption scheme is used to maintain client's input/outpt privacy.
Atallah et al. in \cite{atallah2005} and \cite{atallah2010} explored a list of work in outsourcing computations. In \cite{atallah2005}, a protocol is designed for outsourcing secure sequence comparison using homomorphic encryption techniques. A secure protocol for outsourcing matrix multiplication was presented in \cite{atallah2010} using secret sharing. The wor in \cite{Horvath20, Yang19, shan18} provide
the survey for cryptographic obfuscation and 
secure outsourced computation.
Garg et al. \cite{Garg16} studies candidate indistinguishability obfuscation and functional encryption for all circuits.    
The work in \cite{Cousins18} implements a non-trivial program obfuscation based on polynomial rings. 


The work in \cite{BPSW07} prosents an efficient protocol for privacy-preserving evaluation
of diagnostic programs, represented as binary decision trees or
branching programs. The main purpose of the protocol
is  to maintain the privacy of both the user data and the server's  diagnostic
program. The protocol needs expensive homomorphic encryption and garbled circuits,
so it cannot be applied in complicated  SAT solving.

The work in \cite{CLMTL12} presents a novel secure outsourcing 
algorithm for exponentiation
modular a prime. 
The randomization methods presented in \cite{WRW11} are secure and practical 
methods to randomize LP instances. Works in \cite{Clifton03,Agrawal00}
investigate techniques for privacy preserving data mining.  
These approaches  only apply to problems with known 
computation procedures.  
The work in \cite{Yasin16} presents a SARLock scheme  to enhance
the circuit lock schemes. However, Shamsi
et al. in \cite{Shamsi17} introduce an new version of the SAT attack to defeat the anti-SAT obfuscation schemes such as
SARlock.
The work in \cite{Qin14} proposes an approach to preserve input and output privacy based on CNF obfuscation, and presents obfuscation algorithm and its corresponding solution. However,
the obfuscated formula can be  attacked as demonstrated in \cite{Dimitriou19}. T. Dimitriou presents 
CENSOR, a privacy-preserving obfuscation
for outsourcing SAT formulas in \cite{Dimitriou22}.
At the core of the CENSOR framework lies a mechanism
that transforms any formula to a random one with
the same number of satisfying assignments. 

Many network configuration verification and planning problems can
be converted to SAT.  
Bera et al. \cite{BGD09} presented a framework that  formulates a QSAT
 (satisfiability of quantified boolean formula) based decision problem 
to verify whether the access control  implementation conforms
 to the global policy both in presence and absence 
of the hidden access paths. ConfigChecker \cite{icnp09} models the entire 
network   using
binary decision diagrams (BDDs) \cite{bryant86}, which are  compressed
form of SAT.

\section{Legal Implications for SAT outsourcing}
\label{sec:legal}
There are some legal issues for SAT outsourcing. The purpose of the work
in this paper is to provide privacy for the customers who want to
outsource SAT problems. This may open the door for outsourcing criminal activities
and make the tracking of  criminal activities difficult.
For example, a customer can easily convert the integer factorization problem
to SAT by converting the integer multiplication circuit to a CNF formula.
Then the customer can randomize it and outsource the SAT solving problem
to cloud servers. Though it is unknown for the performance of solving
 integer factoring through SAT transformation, future progress in SAT 
may  provide feasible solutions.  The computation related to 
integer factoring can be directly used for criminal activities, and 
the cloud servers do not know that they are providing service for 
these  activities since the original problems can be randomized. 
It is also very likely that some cloud servers choose not to 
 publicize  efficient algorithms for SAT since good algorithms 
are profitable.  
Thus it is difficult for authorities to track the service provided
by cloud servers and the SAT instances submitted by customers for
criminal investigation.  
We believe that these issues will be important
research topics for cloud computing.

\section{Conclusion}
\label{sec:conclusion}
Outsourcing computations to cloud servers becomes a necessity due to inherit complexity for most of real world problems. SAT outsourcing is important due to broad applications
of SAT. Privacy preserving and information hiding of the original problem can be achieved by 
randomizing SAT instances. In this paper we discussed the importance of SAT outsourcing and how it can be used to randomize computational problems.  We have presented a method to randomize the whole structure of the SAT instance, a method to randomize solution set, and methods to randomize Mincost SAT and MAX3SAT. The evaluation of the presented methods shows that overhead coming from SAT randomization is within the practical limits and it is applicable as shown in the case study. 
For future work, we plan to 
(1) investigate if there exists any other better randomization method, (2) investigate
 the practicality of our approach 
on other applications, and (3) develop an interactive platform for SAT outsourcing.    

\bibliographystyle{plain}
\bibliography{out}

\end{document}